\documentclass[twocolumn]{autart}  
\usepackage{graphicx}      
\usepackage{amsmath}
\usepackage{amsfonts}
\usepackage{apacite}
\usepackage{mathptmx} 
\usepackage{color}
\usepackage{wrapfig}
\graphicspath{{figures/}} 

\newtheorem{theorem}{\indent Theorem}
\newtheorem{lemma}[theorem]{\indent Lemma}
\newtheorem{corollary}[theorem]{\indent Corollary}
\newtheorem{proposition}[theorem]{\indent Proposition}

\newtheorem{remark}[theorem]{\indent Remark}
\newtheorem{assumption}[theorem]{\indent Assumption}
\newenvironment{proof}{\par{\itshape Proof}.\ }

\begin{document}
\begin{frontmatter}

\title{Fault-tolerant $H^\infty$ control for optical parametric oscillators with pumping fluctuations} 
\thanks[support]{This work was supported by the Australian Research Council’s Discovery Projects Funding Scheme under Project DP190101566 and Project DP180101805, the Air Force Office of Scientific Research under Agreement FA2386-16-1-4065, the Centres of Excellence under Grant CE170100012, the Alexander von Humboldt Foundation of Germany, and the U. S. Office of Naval Research Global under Grant N62909-19-1-2129.}
\thanks[cor]{Corresponding Author}

\author[1,2]{Yanan Liu}\ead{yaananliu@gmail.com},    
\author[1]{Daoyi Dong}\ead{daoyidong@gmail.com},               
\author[3]{Ian R. Petersen}\ead{i.r.petersen@gmail.com},  
\author[1,2,cor]{Hidehiro Yonezaw}\ead{h.yonezawa@unsw.edu.au}

\address[1]{School of Engineering and Information Technology, University of New South Wales, ACT 2600, Australia}  
\address[2]{Centre for Quantum Computation and Communication Technology, Australian Research Council, ACT 2600, Australia} 
\address[3]{Research School of Electrical, Energy and Materials Engineering, The Australian National University, ACT 0200, Australia}        

\begin{keyword}                           
Coherent $H^\infty$ control, quantum optics, Riccati equations, fault-tolerant quantum control, quantum controller.           
\end{keyword}

\begin{abstract}                          
Optical Parametric Oscillators (OPOs) have wide applications in quantum optics for generating squeezed states and developing advanced technologies. When the phase or/and the amplitude of the pumping field for an OPO have fluctuations due to fault signals, time-varying uncertainties will be introduced in the dynamic parameters of the system. In this paper, we investigate how to design a fault-tolerant $H^\infty$ controller for an OPO with a disturbance input and time-varying uncertainties, which can achieve the required $H^\infty$ performance of the quantum system. We apply robust $H^\infty$ control theory to a quantum system, and design a passive controller and an active controller based on the solutions to two Riccati equations. The passive controller has a simple structure and is easy to be implemented by using only passive optical components, while the active quantum controller may achieve improved performance. The control performance of the proposed two controllers and one controller that was designed without consideration of system uncertainties is compared by numerical simulations in a specific OPO, and the results show that the designed controllers work effectively for fluctuations in both the phase and amplitude of the pumping field.
\end{abstract}

\end{frontmatter}

\section{Introduction}
Quantum control theory mainly focuses on developing effective and reliable control methods for quantum systems, and plays a fundamental role in many quantum technologies \cite{altafini2012modeling,xiang2017coherent,guo2018optimal,gao2016fault,4797787,li2006control,ge2020robust,dong2019learning}. Quantum systems inevitably suffer from fault processes in many practical applications, which reduces the performance of systems and even causes the failure of the system \cite{blanke2006diagnosis}. For example, the amplitude and phase of the pumping field in an Optical Parametric Oscillator (OPO) may be subject to fluctuations due to unstable laser power. Fault-tolerant control theory aims to detect the faults and design a controller to make the system fault tolerant \cite{ding2008model}, which has been extended to the quantum domain \cite{liu2020fault,wang2016fault}.

An OPO is essentially composed of an optical resonator and a nonlinear optical crystal, and has been widely used in quantum optical experiments to generate squeezed states \cite{takeno2007observation,serikawa2016creation,suzuki20067}. The dynamics of an OPO with an ideal pumping field, by which we mean a laser with fixed amplitude and phase, can be described by a time-invariant linear differential equation \cite{liu2020coherent,bachor2019guide}. However, when the phase or/and the amplitude of the pumping field are subject to fluctuations due to fault processes, they will introduce time-varying uncertainties in the system parameters \cite{gao2016fault}. In this work, we aim to design a feedback controller to stabilise an OPO with a disturbance input and parameter uncertainties. Since OPOs have been widely used to generate squeezed light in quantum optics, the proposed robust $H^\infty$ control theory can improve the robustness and reliability of the OPOs, thus having the potential to enhance the performance of quantum technologies such as quantum sensing, computation and communication where squeezed states are utilized.

Compared with measurement-based feedback control, coherent feedback control can avoid quantum measurement and feedback delay \cite{wiseman2009quantum,altafini2012modeling,dong2010quantum,liu2016lyapunov,liu2019feedback,yamamoto2007feedback,iida2012experimental,zhang2010direct,yamamoto2008avoiding}. Many feedback control methods including LQG control and $H^\infty$ control have been used to design coherent feedback controllers and to improve the performance of quantum optical systems \cite{james2008h,nurdin2009coherent,maalouf2010coherent,xiang2017coherent}. For example, the design of $H^\infty$ quantum controllers for a time-invariant linear quantum system was proposed in \cite{james2008h}. However, when there is parameter uncertainty in the plant modeling, the standard $H^\infty$ theory cannot provide guaranteed $H^\infty$ performance as well as stability of the closed-loop systems \cite{xu2002robust}. This has motivated the study of robust $H^\infty$ control problems. The objective of robust $H^\infty$ control theory is to obtain a controller such that $H^\infty$ norm of an input-output operator is no larger than a prescribed level for all admissible uncertainties. In this work, the robust $H^\infty$ control theory will be applied to an OPO system. 

The rest of this paper is organized as follows. Section \ref{sec2} presents the system model and the problem formulation. In Section \ref{sec3}, a passive quantum controller is designed by applying $H^\infty$ control theory to the OPO with phase fluctuations, where the controller is implemented by an empty cavity. In Section \ref{sec4}, an active quantum controller is designed based on another decomposition of the system uncertainties. In Section \ref{sec5}, we analyse the effectiveness of the two proposed controllers in an OPO with fluctuations in both the amplitude and phase of the pumping field. Section \ref{sec:conclusion} concludes this paper.

\section{System Models and Problem Formulation}
\label{sec2}
Continuous wave (Cw)-pumped OPOs have been widely used to obtain squeezed light in many experiments \cite{takeno2007observation,serikawa2016creation}. A pumping laser is often assumed to have fixed amplitude and phase to generate squeezed states. However, the amplitude and phase may be affected by classical fault processes. For example, the laser output power may drift for a long time operation \cite{gao2016fault}. In this section, we consider an OPO composed of three mirrors, where the phase of the pumping field is subject to fluctuations. The structure of the OPO is shown in Fig. \ref{fig:plant}.

\begin{figure}
\includegraphics[width=0.5\textwidth]{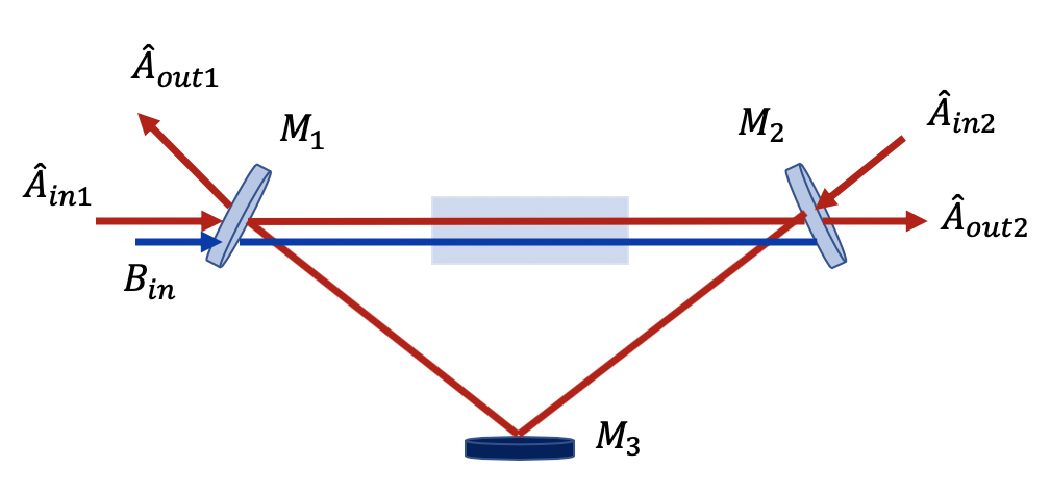}
\caption{Structure of an OPO, composed of three mirrors $M_1$, $M_2$, $M_3$, and a nonlinear crystal inside the cavity. $M_1$ and $M_2$ are partially transmissive for the fundamental field $\hat{A}_{in1}$ and $\hat{A}_{in2}$, and fully transmissive for pumping field $B_{in}$, and $M_3$ is fully reflective for both fields.\label{fig:plant}}
\end{figure}
Let $\hat{a}(t)$ and $\hat{a}^{\dagger}(t)$ denote the annihilation and creation operators of the fundamental field inside the cavity. Then their dynamical equations can be written as follows:
\begin{equation}
\label{eqn:differentialofannihilation}
\begin{aligned}
d\hat{a}(t)&=\left[-\frac{\kappa}{2}\hat{a}(t)+\chi^{(2)}\beta e^{\mbox{i}[\phi+\Delta\phi(t)]}\hat{a}^{\dagger}(t)\right]dt\\
&+\sqrt{\kappa_1}d\hat{A}_{in1}(t)+\sqrt{\kappa_2}d\hat{A}_{in2}(t),\\
d\hat{a}^{\dagger}(t)&=\left[-\frac{\kappa}{2}\hat{a}^{\dagger}(t)+\chi^{(2)}\beta e^{-\mbox{i}[\phi+\Delta\phi(t)]}\hat{a}(t)\right]dt\\
&+\sqrt{\kappa_1}d\hat{A}_{in1}^{\dagger}(t)+\sqrt{\kappa_2}d\hat{A}_{in2}^{\dagger}(t).
\end{aligned}
\end{equation}
Here, $\chi^{(2)}$ is the nonlinearity of the crystal; $\mbox{i}=\sqrt{-1}$; the pumping field $B_{in}$ in Fig. \ref{fig:plant} is represented by a complex number $\beta e^{\mbox{i}[\phi+\Delta\phi(t)]}$, where $\beta\in\mathbb{R}$ is the amplitude and $\phi+\Delta \phi(t)$ is the phase with a fluctuation $\Delta \phi(t)$ caused by a classical fault process. Without loss of generality, we let $\phi=0$ and consider the effect of phase fluctuation $\Delta \phi(t)$; $\kappa=\kappa_1+\kappa_2$, where $\kappa_1$ and $\kappa_2$ are the decay rates for mirrors $M_1$ and $M_2$; $\kappa_i=\frac{T_ic}{l}$, where $T_i$ is the power transmissivity of the mirror $M_i$, $c$ is the speed of light, and $l$ is the optical length of the cavity. $\hat{A}_{in1}$ and $\hat{A}_{in2}$ are two input fields.

The equations of two output fields $\hat{A}_{out1}$ and $\hat{A}_{out2}$ are:
\begin{equation}
\label{eqn:outputfields}
\begin{aligned}
\hat{A}_{out1}(t)&=\sqrt{\kappa_1}\hat{a}(t)-\hat{A}_{in1}(t),\\
\hat{A}_{out2}(t)&=\sqrt{\kappa_2}\hat{a}(t)-\hat{A}_{in2}(t).
\end{aligned}
\end{equation}
Since the annihilation and creation operators are not self-adjoint, we consider the amplitude and phase quadrature of the fundamental field as the system variables: $$x=\begin{bmatrix}
\hat{a}(t)+\hat{a}^{\dagger}(t)\\
-\mbox{i}(\hat{a}(t)-\hat{a}^{\dagger}(t))
\end{bmatrix}.$$
We denote
\[
\begin{aligned}
\omega(t)&=\left[\begin{smallmatrix}
\hat{A}_{in1}(t)+\hat{A}_{in1}^{\dagger}(t)\\
-\mbox{i}(\hat{A}_{in1}(t)-\hat{A}_{in1}^{\dagger}(t))
\end{smallmatrix}\right],u(t)=\left[\begin{smallmatrix}
\hat{A}_{in2}(t)+\hat{A}_{in2}^{\dagger}(t)\\
-\mbox{i}(\hat{A}_{in2}(t)-\hat{A}_{in2}^{\dagger}(t))
\end{smallmatrix}\right],\\
z(t)&=\left[\begin{smallmatrix}
\hat{A}_{out2}(t)+\hat{A}_{out2}^{\dagger}(t)\\
-\mbox{i}(\hat{A}_{out2}(t)-\hat{A}_{out2}^{\dagger}(t))
\end{smallmatrix}\right],y(t)=\left[\begin{smallmatrix}
\hat{A}_{out1}(t)+\hat{A}_{out1}^{\dagger}(t)\\
-\mbox{i}(\hat{A}_{out1}(t)-\hat{A}_{out1}^{\dagger}(t))
\end{smallmatrix}\right].
\end{aligned}
\]

Then, the system equations can be written as follows:
\begin{equation}
\label{eqn:differentialofquadratures}
\begin{aligned}
dx(t)&=[A+\Delta A(t)]x(t)dt+B_1du(t)+B_2d\omega(t),\\
z(t)&=C_1 x(t)+D_1u(t),\\
y(t)&=C_2 x(t)+D_2\omega(t),
\end{aligned}
\end{equation}
where
\begin{equation}
\label{eqn:systemparameters}
\begin{aligned}
& A+\Delta A(t)=\left[\begin{smallmatrix}
-\frac{\kappa}{2} & 0\\
0 & -\frac{\kappa}{2}
\end{smallmatrix}\right]+\chi^{(2)}\beta \left[\begin{smallmatrix}\cos (\Delta \phi(t))&\sin (\Delta \phi(t))\\\sin (\Delta \phi(t))&-\cos (\Delta \phi(t)) \end{smallmatrix}\right],\\
& B_1=\sqrt{\kappa_{2}}\mbox{I}, B_2=\sqrt{\kappa_{1}}\mbox{I},\\
& C_1=\sqrt{\kappa_{2}}\mbox{I}, D_1=-\mbox{I},\\
& C_2=\sqrt{\kappa_1}\mbox{I}, D_2=-\mbox{I}.
\end{aligned}
\end{equation}

Here, $\mbox{I}$ is the identity matrix with an appropriate dimension. The initial system variables $x(0)$ usually consist of quadrature operators satisfying the commutation relations \cite{james2008h}:
\begin{equation}
[x_j(0), x_k(0)]=2\mbox{i}\Theta_{jk}, j, k=1, \cdots, n,
\end{equation}
where $\Theta$ is a real anti-symmetric matrix with components $\Theta_{jk}$ and $n$ is the system dimension. For simplicity, we consider a canonical form: $\Theta={\rm diag}(J, J, \cdots, J)$, where $J$ denotes a real skew-symmetric $2\times 2$ matrix $J=\left[\begin{smallmatrix}0&1\\-1&0\end{smallmatrix}\right]$.
In our case, we have $\Theta=J$ for $n=2$. Also, $u(t)$ is the control signal to be designed. In addition, $d\omega(t)=\beta_\omega(t)dt+d\tilde{\omega}(t)$ represents a disturbance input, where $\beta_\omega(t)$ is a self adjoint process and $\tilde{\omega}(t)$ is the quantum noise part. Furthermore, $y(t)$ is the output that is fed back to the controller and $z(t)$ is the controlled output. Also, $\Delta A(t)$ is a real-valued matrix function representing time-varying parameter uncertainties. 

We first give the following lemmas to illustrate the stable and quadratically stable systems with disturbance attenuation $\gamma$ for systems described by linear differential equations \eqref{eqn:differentialofquadratures}.
\begin{lemma}{\cite{xie1992h}}
Given a scalar $\gamma>0$, the following system
\begin{equation}
\begin{aligned}
dx(t)&=Ax(t)dt+B_1d\omega(t),\\
z(t)&=C_1x(t).\\
\end{aligned}
\end{equation}
is stable with disturbance attenuation $\gamma$ if it satisfies the following conditions:
\begin{itemize}
\item [i)] $A$ is a stable matrix; 
\item [ii)] the transfer function from disturbance $\omega$ to controlled output $z$ satisfies
\begin{equation*}
\|C_1(s\mbox{I}-A)^{-1}B_1\|_{\infty}<\gamma.
\end{equation*}
\end{itemize}
\end{lemma}

\begin{lemma}{\cite{xie1992h}}
Given a $\gamma>0$, and a system with parameter uncertainty $\Delta A(t)$ in the state matrix as:
\begin{equation}
\label{eqn:syswithuncertaintyinA}
\begin{aligned}
dx(t)&=[A+\Delta A(t)]x(t)dt+B_1d\omega(t),\\
z(t)&=C_1x(t).\\
\end{aligned}
\end{equation}
This uncertain system is quadratically stable with disturbance attenuation $\gamma$ if there exists a positive definite symmetric matrix $P$ such that for all admissible uncertainties $\Delta A(t)$
\begin{equation}
\begin{aligned}
\left[A+\Delta A(t)\right]^T P+P[A+\Delta A(t)]&+\gamma^{-2}PB_1B_1^TP\\
&+C_1^TC_1<0.
\end{aligned}
\end{equation}
\end{lemma}

Furthermore, the uncertain system \eqref{eqn:differentialofquadratures} is said to be quadratically stabilisable with disturbance attenuation $\gamma$ via a linear dynamical controller if there exists a linear dynamical feedback compensator $K(s)$ such that with $u=K(s)y$ the resulting closed-loop system is quadratically stable with disturbance attenuation $\gamma$.

The fluctuations of the pumping phase in an OPO will reduce its performance, and may even make the system unstable. Hence, the control objective in this paper is to design an $H^\infty$ feedback controller such that the closed-loop system corresponding to the system \eqref{eqn:differentialofquadratures} and the designed controller is quadratically stable and achieves a given level of disturbance attenuation $\gamma$ for all norm-bounded uncertainties. In particular, we consider two schemes: one is to design a passive controller and the other is to employ an active controller.

\section{Passive Quantum Controller Design}
\label{sec3}
In this section, we design a passive quantum controller to achieve the given disturbance attenuation $\gamma$ for the OPO shown in Fig. \ref{fig:plant}. Here, a passive controller means that the controller is composed by only passive components like cavities, beam-splitters, and phase shifters \cite{petersen2011cascade}.

\subsection{Robust $H^\infty$ controller design}
\label{sec:controldesigntheory}
We first assume that the controller to be designed is described by the following dynamical equation:
\begin{equation}
\label{eqn:controller}
\begin{aligned}
dx_c(t)&=\mathcal{A}_cx_c(t)dt+\mathcal{B}_cdy(t),\\
u(t)&=\mathcal{C}_cx_c(t)dt.
\end{aligned}
\end{equation}
In coherent feedback control, the final controller may need to be modified based on \eqref{eqn:controller} by adding other additional parameters such that this dynamical equation represents a physical realizable quantum controller.

If the uncertainty $\Delta A(t)$ is norm-bounded and is of the form
\begin{equation}
\label{eqn:uncertaintydecompositionA}
\Delta A(t)=H_1F(t)E_1,
\end{equation}
where $H_1$ and $E_1$ are known constant matrices and $F(t)$ is an unknown matrix function satisfying
\begin{equation}
\label{eqn:normboundedcondition}
F^T(t)F(t)\leq \rho^2\mbox{I},
\end{equation}
with a constant $\rho>0$, referred to an uncertainty bound.

The following theorem solves the robust $H^\infty$ control design problem by solving two Riccati equations.
\begin{theorem}
\label{the:Riccati}
Let $\gamma$ be a required level of disturbance attenuation. The system \eqref{eqn:differentialofquadratures} is quadratically stable with $\gamma$ via the controller \eqref{eqn:controller} if and only if there exists a constant $\epsilon>0$ such that the Riccati equations \eqref{eqn:simRiccati_a} and \eqref{eqn:simRiccati_b} have stabilising solutions $X, Y$ to ensure that the spectral radius satisfies $\zeta(XY)<1$:
\begin{subequations}
\begin{align}
&\left(A-B_1G^{-1}D_1^TC_1\right)^TX+X\left(A-B_1G^{-1}D_1^TC_1\right)\label{eqn:simRiccati_a}\notag\\
&-X(B_1G^{-1}B_1^T-\epsilon\rho^2H_1H_1^T-\gamma^{-2}B_2B_2^T)X\notag\\
&+C_1^T(\mbox{I}-D_1G^{-1}D_1^T)C_1+\frac{1}{\epsilon}E_1^TE_1=0,\\
&(A-B_2D_2^T\Gamma^{-1}C_2)Y+Y(A-B_2D_2^T\Gamma^{-1}C_2)^T\label{eqn:simRiccati_b}\notag\\
&-Y(\gamma^2C_2^T\Gamma^{-1}C_2-\frac{1}{\epsilon}E_1^TE_1-C_1^TC_1)Y\notag\\
&+\gamma^{-2}B_2(I-D_2^T\Gamma^{-1}D_2)B_2^T+\epsilon\rho H_1H_1^T=0.
\end{align}
\end{subequations}
Using the solutions $X$ and $Y$, the controller is given by
\begin{equation}
\label{eqn:simcontrolparameters}
\begin{aligned}
\mathcal{A}_c&=A+B_1\mathcal{C}_c-\mathcal{B}_cC_2\\
&+[\epsilon \rho^2H_1H_1^T+\gamma^{-2}(B_2-\mathcal{B}_cD_2)B_2^T]X,\\
\mathcal{B}_c&=(\mbox{I}-YX)^{-1}(YC_2^T+\gamma^{-2}B_2D_2^T)\Gamma^{-1},\\
\mathcal{C}_c&=-G^{-1}(B_1X+D_1^TC_1).
\end{aligned}
\end{equation}
Here, $G=D_1^TD_1$, $\Gamma=D_2D_2^T$.
\end{theorem}

\begin{proof}
For a system with time-varying uncertainties, which is described as:
\begin{equation}
\label{eqn:syswithgeneraluncertainty}
\begin{aligned}
dx(t)&=[A+\Delta A(t)]x(t)dt+B_1du(t)+[B_2+\Delta B(t)]d\omega(t),\\
z(t)&=C_1x(t)+D_1u(t),\\
y(t)&=[C_2+\Delta C(t)]x(t)+D_{21}\omega(t)+[D_{22}+\Delta D(t)]u(t),
\end{aligned}
\end{equation}
assume that the uncertainties are norm-bounded and of the form
\begin{equation}
\label{eqn:uncertaintydecomposition}
\begin{bmatrix}\Delta A(t)&\Delta B(t)\\ \Delta C(t)&\Delta D(t)\end{bmatrix}=\begin{bmatrix}H_1\\H_2\end{bmatrix}F(t)\begin{bmatrix}E_1&E_2\end{bmatrix},
\end{equation}
where $H_1, H_2, E_1$, and $E_2$ are known constant matrices and $F(t)$ is an unknown matrix function satisfying \eqref{eqn:normboundedcondition}.

Theorem 3.1 in \cite{xie1992h} proves that the system \eqref{eqn:syswithgeneraluncertainty} is quadratically stabilisable with the given $\gamma$ via a controller of the form \eqref{eqn:controller} if and only if the following corresponding deterministic system is stabilisable with the unitary disturbance attenuation via the same controller \eqref{eqn:controller}:
\begin{equation}
\label{eqn:detersystem}
\begin{aligned}
dx(t)&=Ax(t)+B_1du(t)+\begin{bmatrix}\sqrt{\epsilon}\rho H_1&\gamma^{-1}B_2\end{bmatrix}d\omega^\prime(t),\\
\tilde{z}(t)&=\begin{bmatrix}\frac{1}{\sqrt{\epsilon}}E_1\\C_1\end{bmatrix}x(t)+\begin{bmatrix}\frac{1}{\sqrt{\epsilon}}E_2\\D_1\end{bmatrix}u(t),\\
y(t)&=C_2x(t)+\begin{bmatrix}\sqrt{\epsilon}\rho H_2&\gamma^{-1}D_2\end{bmatrix}\omega^\prime(t),
\end{aligned}
\end{equation}
where $\omega^\prime$ is the disturbance input.

Now, the problem is to design a scaled $H^\infty$ controller for the system \eqref{eqn:detersystem} such that the closed-loop system is stable with unitary disturbance attenuation $g=1$. Applying the scaled $H^\infty$ control result from \cite{petersen1991first}, it is straightforward to obtain \eqref{eqn:simRiccati_a} and \eqref{eqn:simRiccati_b}. This completes the proof.
\end{proof}

\subsection{Existence of solutions to Riccati equations}
\label{sec}

According to Theorem \ref{the:Riccati}, the design of a robust $H^\infty$ controller is related to the solutions to Riccati equations \eqref{eqn:simRiccati_a} and \eqref{eqn:simRiccati_b}. Hence, in this section, we analyse the necessary and sufficient conditions on the existence of solutions to the Riccati equations \eqref{eqn:simRiccati_a} and \eqref{eqn:simRiccati_b}.

For the uncertainties of the OPO described in \eqref{eqn:differentialofquadratures}, if we denote 
$$A=\begin{bmatrix}-\frac{\kappa}{2}&0\\0& -\frac{\kappa}{2}\end{bmatrix},$$
we can obtain 
\begin{equation}
\label{eqn:systemuncertainty_1}
    \Delta A_1(t)=\chi \begin{bmatrix}\cos (\Delta \phi(t))&\sin (\Delta \phi(t))\\\sin (\Delta \phi(t))&-\cos (\Delta \phi(t))\end{bmatrix},
\end{equation}
where $\chi=\chi^{(2)}\beta$. 
Therefore, we may choose $H_1=\chi\mbox{I}$, $E_1=\mbox{I}$, and obtain
$$F_1(t)=\begin{bmatrix}\cos (\Delta \phi(t))&\sin (\Delta \phi(t))\\\sin (\Delta \phi(t))&-\cos (\Delta \phi(t)) \end{bmatrix}.
$$  
The norm-bounded condition can be ensured by the following equation
\begin{equation*}
F_1^T(t)F_1(t)=\mbox{I}\leq \rho_1^2 \mbox{I},
\end{equation*}
with $\rho_1=1$. Then we get the following proposition about the solutions to Riccati equations \eqref{eqn:simRiccati_a} and \eqref{eqn:simRiccati_b}.

\begin{proposition}
\label{existenceconditions_case1}
For the OPO in Fig.\ref{fig:plant} with parameters shown in \eqref{eqn:systemparameters} and the system uncertainty shown in \eqref{eqn:systemuncertainty_1}. If $\kappa_2 - \frac{\kappa_1}{\gamma^2 }>0$, we can find a positive $\epsilon>0$ such that
\begin{equation}
\label{con:controlable1_1}
\kappa_2 - \frac{\kappa_1}{\gamma^2 }- \chi^2\epsilon\rho_1^2> 0,
\end{equation}
and \eqref{eqn:simRiccati_a} has a stabilising solution $X$ if and only if:
\begin{equation}
\label{con:nopureimaginaryeigen1_1}
\epsilon\gamma^2\left[(\kappa_1-\kappa_2)^2-4\chi^2\rho_1^2\right]+4(\gamma^2\kappa_2-\kappa_1)\geq0.
\end{equation}
If $\kappa_2 - \frac{\kappa_1}{\gamma^2 }\leq 0$, we can find a positive $\epsilon>0$ such that
\begin{equation}
\label{con:controlable1_2}
\kappa_1\gamma^2 - \kappa_2 -\frac{1}{\epsilon}\geq 0,
\end{equation}
the Riccati equation \eqref{eqn:simRiccati_b} has a stabilising solution $Y$ if and only if:
\begin{equation}
\label{con:nopureimaginaryeigen1_2}
4\chi^2\rho_1^2\epsilon(\gamma^2\kappa_1-\kappa_2)+(\kappa_1-\kappa_2)^2-4\chi^2\rho_1^2\geq 0.
\end{equation}
\end{proposition}
\begin{proof}
Theorem 14.16 in \cite{levine2018control} provides a sufficient and necessary condition for the solutions to Riccati equation in the following form
\begin{equation}
\label{eqn:generalriccati}
A^TX+XA-XBX+C=0,
\end{equation}
under the condition that $B=B^T, B\geq 0$ and $C=C^T, C\geq 0$. In particular, Theorem 14.16 in \cite{levine2018control} states that the Riccati equation \eqref{eqn:generalriccati} has a stabilising solution if and only if $(A, B)$ is controllable and the Hamiltonian matrix $H=\begin{bmatrix}A&-B\\-C&-A^T\end{bmatrix}$ has no pure imaginary eigenvalues. We first consider the semi-positive definite of $B_1G^{-1}B_1^T-\epsilon \rho^2H_1H_1^T-\gamma^{-2}B_2B_2^T$ for the first Riccati equation \eqref{eqn:simRiccati_a}, we obtain
\begin{equation}
\label{eqn:semipositiveB1}
    \kappa_2 - \frac{\kappa_1}{\gamma^2 }- \chi^2\epsilon\rho_1^2\geq 0,
\end{equation}
only when $\kappa_2-\gamma^{-2}\kappa_1>0$ can we find a positive $\epsilon$ to ensure that \eqref{eqn:semipositiveB1} holds.

To ensure the controllability of 
$$(A-B_1G^{-1}D_1^TC_1, B_1G^{-1}B_1^T-\epsilon \rho^2H_1H_1^T-\gamma^{-2}B_2B_2^T),$$ in the first Riccati equation \eqref{eqn:simRiccati_a}, we first need to choose a proper $\epsilon$ such that the following two equations are not satisfied simultaneously: 
\begin{equation}\label{eqn:intercontrollable1_1}
\begin{aligned}
\kappa_2 - \frac{\kappa_1}{\gamma^2 }- \chi^2\epsilon\rho_1^2&=0,\\
\frac{\kappa_1-\kappa_2}{2}\left(\frac{\kappa_1}{\gamma^2} -\kappa_2 + \chi^2\epsilon\rho_1^2\right)&=0,
\end{aligned}
\end{equation}
which means 
\begin{equation}
\label{eqn:controllability1}
    \kappa_2-\frac{\kappa_1}{\gamma^2}-\chi^2\epsilon\rho_1\neq 0.
\end{equation}
From \eqref{eqn:controllability1} and \eqref{eqn:semipositiveB1}, we obtain \eqref{con:controlable1_1} when $\kappa_2-\frac{\kappa_1}{\gamma^2}>0$.

From \eqref{con:controlable1_1}, we ensure that $B_1G^{-1}B_1^T-\epsilon \rho^2H_1H_1^T-\gamma^{-2}B_2B_2^T$ is semi-positive definite, and $$(A-B_1G^{-1}D_1^TC_1, B_1G^{-1}B_1^T-\epsilon \rho^2H_1H_1^T-\gamma^{-2}B_2B_2^T)$$ is controllable. We are now in the position to ensure that the Hamiltonian matrix $H$ of the Riccati equation \eqref{eqn:simRiccati_a} has no pure imaginary eigenvalues, which results in the inequality \eqref{con:nopureimaginaryeigen1_1}.

Similarly, when $\kappa_2-\frac{\kappa_1}{\gamma^2}<0$, we can find a positive $\epsilon$ to obtain
\begin{equation}
\label{eqn:positiveC1}
    \kappa_1\gamma^2 - \kappa_2 -\frac{1}{\epsilon}\geq 0,
\end{equation}
which ensures $\gamma^2C_2^T\Gamma^{-1}C_2-\frac{1}{\epsilon}E_1^TE_1-C_1^TC_1\geq 0$ in the second Riccati equation \eqref{eqn:simRiccati_a}.

From the controllability of 
\begin{equation*}
(A-B_2D_2^T\Gamma^{-1}C_2, \gamma^2C_2^T\Gamma^{-1}C_2-\frac{1}{\epsilon}E_1^TE_1-C_1^TC_1),
\end{equation*}
we obtain that the following two inequalities are not satisfied simultaneously:
\begin{equation}\label{eqn:intercontrollable1_2}
\begin{aligned}
\kappa_1\gamma^2 - \kappa_2 -\frac{1}{\epsilon}&=0,\\
-\frac{\kappa_1-\kappa_2}{2}\left(\kappa_2 - \gamma^2\kappa_1 +\frac{1}{\epsilon}\right)&=0,
\end{aligned}
\end{equation}
which means 
\begin{equation}
\label{eqn:controlB2}
    \kappa_1\gamma^2-\kappa_2-\frac{1}{\epsilon}\neq 0.
\end{equation}
From \eqref{eqn:controlB2} and \eqref{eqn:positiveC1}, we can find a positive $\epsilon$ to obtain \eqref{con:controlable1_2} when $\kappa_2-\frac{\kappa_1}{\gamma^2}\leq 0$. Then we turn to consider the eigenvalues of the Hamiltonian matrix for the second Riccati \eqref{eqn:simRiccati_b}, from which the inequality \eqref{con:nopureimaginaryeigen1_2} is obtained. Thus, we complete the proof.
\end{proof}

\begin{remark}
When $\kappa_2-\frac{\kappa_1}{\gamma^2}>0$, we can obtain the sufficient and necessary conditions of the solution existence to the first Riccati equation \eqref{eqn:simRiccati_a}. While in this case Theorem 14.16 in \cite{levine2018control} cannot be applied to the second Riccati equation due to the negativity of $\gamma^2C_2^T\Gamma^{-1}C_2-\frac{1}{\epsilon}E_1^TE_1-C_1^TC_1$. However, it is still possible to obtain a stabilising solution to the second Riccati equation in some cases (see the numerical example in Sec. \ref{sec:numericalcalculationcase1}). We have a similar conclusion when $\kappa_2-\frac{\kappa_1}{\gamma^2}\leq 0$. The conditions we obtained in Proposition \ref{existenceconditions_case1} are only sufficient and necessary for one of the two Riccati equations \eqref{eqn:simRiccati_a} and \eqref{eqn:simRiccati_b}.
\end{remark}

\subsection{Numerical example for passive controller design}
\label{sec:numericalcalculationcase1}
In this section, a numerical example of an OPO with specific parameters is used to show the design procedure for a passive coherent controller \cite{serikawa2016creation}, the simulation results are obtained by using MATLAB 2020b. By taking the values of transmissivity $T_i$ and of the optical length $l$ from experiments in \cite{serikawa2016creation}, we calculate the decay rates for mirrors $M_1$ and $M_2$ as $\kappa_1=\frac{cT_1}{l}=0.0011$, $\kappa_2=\frac{cT_2}{l}=0.8264$, $\kappa=\kappa_1+\kappa_2=0.8275$, and take the pumping-related coefficient as $\chi=\chi^{(2)}\beta=0.1\frac{\kappa}{2}=0.0414$.

Before we design the controller, a positive $\epsilon$ needs to be chosen such that the two Riccati equations \eqref{eqn:simRiccati_a} and \eqref{eqn:simRiccati_b} have stabilising solutions. The range of possible $\epsilon$ changes with different disturbance attenuation $\gamma$ and $\rho_1$, which is shown in Fig. \ref{fig:case1possibleepsilon}.
\begin{figure}
\includegraphics[width=0.5\textwidth]{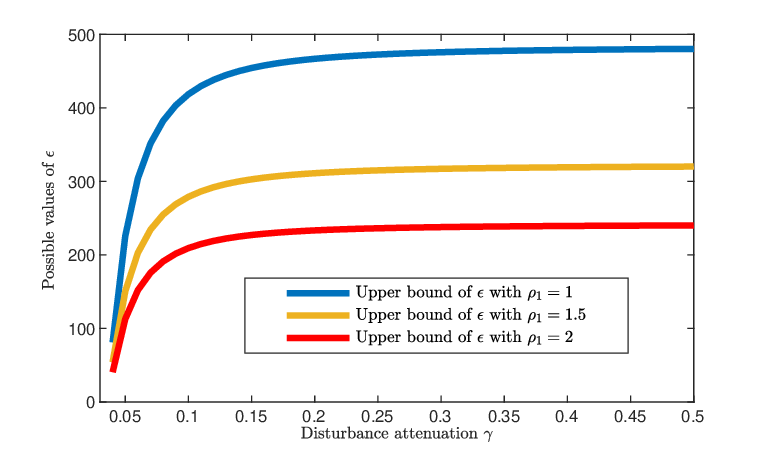}
\caption{Ranges of $\epsilon$ change with disturbance attenuation $\gamma$ for different $\rho_1$.}
\label{fig:case1possibleepsilon}
\end{figure}

In Fig. \ref{fig:case1possibleepsilon}, three different solid lines represent the upper bounds on $\epsilon$ for different $\rho_1$ based on the condition \eqref{con:controlable1_1}. Also a positive $\epsilon$ should be guaranteed.
The value of $\epsilon$ for a given $\gamma$ and $\rho_1$ can be chosen within the area between the upper bound and the lower bound, which is shown as the horizontal axis in the Fig. \ref{fig:case1possibleepsilon}.

For a given disturbance attenuation $\gamma=0.05$, we can choose $\epsilon=1$ to ensure the existence of solutions to Riccati equations based on Fig. \ref{fig:case1possibleepsilon}. Then we obtain the system parameters as:
\begin{equation} 
\begin{aligned}
A&=\begin{bmatrix}-0.4138&0\\0&-0.4138\end{bmatrix},\\
F_1&=\begin{bmatrix}\cos(\Delta \phi(t))&\sin(\Delta \phi(t))\\ \sin(\Delta \phi(t))&-\cos(\Delta \phi(t))\end{bmatrix},\\
B_1&=\begin{bmatrix}0.9091&0\\0&0.9091\end{bmatrix}, B_2=\begin{bmatrix}0.0332&0\\0&0.0332\end{bmatrix},\\
C_1&=\begin{bmatrix}0.9091&0\\0&0.9091\end{bmatrix}, C_2=\begin{bmatrix}0.0332&0\\0&0.0332\end{bmatrix},\\
D_1&=-\mbox{I}, D_2=-\mbox{I}.
\end{aligned}
\end{equation}
When the phase fluctuations satisfy $\Delta \phi(t)\in [-\pi, \pi]$, we have $\rho_1=1$, and can obtain solutions of the Riccati equations \eqref{eqn:simRiccati_a} and \eqref{eqn:simRiccati_b} as follows:
\begin{equation*}
X=\begin{bmatrix}3.0092&0\\0&3.0092\end{bmatrix}, Y=\begin{bmatrix}0.0021&0\\0&0.0021\end{bmatrix}.\\
\end{equation*}
The controller parameters are designed as:
\begin{equation}
\label{eqn:controller1parameters}
\begin{aligned}
\mathcal{A}_c&=\left[\begin{smallmatrix} -2.0763&0\\0&-2.0763 \end{smallmatrix}\right],\mathcal{B}_c=\left[\begin{smallmatrix} -0.0334&0\\0&-0.0334 \end{smallmatrix}\right],\\
\mathcal{C}_c&=\left[\begin{smallmatrix} -1.8265&0\\0&-1.8265 \end{smallmatrix}\right].
\end{aligned}
\end{equation}

For the designed controller given by \eqref{eqn:controller1parameters}, its physical realizability can be checked by the following Lemma.
\begin{lemma}\cite{james2008h}
\label{lemma:physicalrealisation}
The system described by $\{A, B, C\}$ is physical realisable if and only if:
\begin{equation}
\label{eqn:physicalrealisation}
\begin{aligned}
&\mbox{i}A\Theta+\mbox{i}\Theta A^T+BT_\omega B^T=0,\\
&B\left[\begin{smallmatrix}I\\0\end{smallmatrix}\right]=\Theta C^TP^T\times \left[\begin{smallmatrix}0&I\\-I&0\end{smallmatrix}\right]P=\Theta C^T{\rm diag}(J).
\end{aligned}
\end{equation}
Here, $P_m$ is a $2m\times 2m$ permutation matrix, for an arbitrary vector $a=\begin{bmatrix}a_1& a_2& \cdots &a_{2m}\end{bmatrix}^T$, then $P_m a=\begin{bmatrix}a_1 &a_3& \cdots& a_{2m-1}& a_2& a_4& \cdots &a_{2m}\end{bmatrix}^T$.
\end{lemma}

However, most controllers designed by the proposed $H^\infty$ control method do not satisfy the physical realizability conditions in Lemma \ref{lemma:physicalrealisation}. To achieve the coherent feedback control design, we will first apply the algorithm proposed in \cite{vuglar2016quantum} to introduce additional quantum noises such that the controller is physical realisable. Then we may implement the controller by using some passive optical components. In particular, for the control parameters in \eqref{eqn:controller1parameters}, which do not satisfy the physical realizability condition \eqref{eqn:physicalrealisation}, it is possible to introduce additional quantum noises, by which we may reconstruct the input matrix $\begin{bmatrix}\mathcal{B}_c&B_{\nu1}&B_{\nu2}\end{bmatrix}$ to ensure that \eqref{eqn:physicalrealisation} holds. According to the algorithm from \cite{vuglar2016quantum}, we can calculate the input matrices $B_{\nu1}, B_{\nu2}$ in terms of the additional noises $\nu_1$ and $\nu_2$ as
\begin{equation}
B_{\nu1}=\begin{bmatrix}1.8265&0\\0&1.8265\end{bmatrix}, B_{\nu2}=\begin{bmatrix}0.9029&0\\0&0.9029\end{bmatrix}.
\end{equation}
It can be verified that the controller $\{\mathcal{A}_c, \begin{bmatrix}\mathcal{B}_c&B_{\nu1}&B_{\nu2}\end{bmatrix}, \mathcal{C}_c\}$ satisfies \eqref{eqn:physicalrealisation}. The controller can be implemented by an empty cavity with three mirrors, and the decay rates of $M_1^{\prime}$, $M_2^{\prime}$ and $M_3^{\prime}$ are $\kappa_1=0.0011$, $\kappa_2=3.3361$, and $\kappa_3=0.8152$. The structure of the passive controller is shown in Fig. \ref{fig:controller1}.
\begin{figure}
\includegraphics[width=0.5\textwidth]{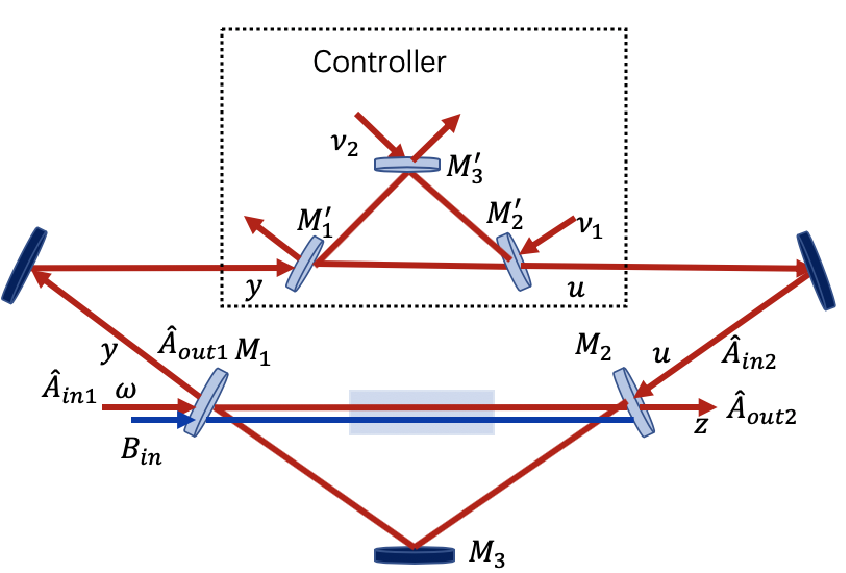}
\caption{The connection of the designed passive quantum controller and the plant, where the controller is shown in dashed square. The passive controller is composed by three mirrors $M_1^{\prime}, M_2^{\prime}$, and $M_3^{\prime}$, which are partially transmissive for the field. $\nu_1$ and $\nu_2$ are two additional quantum noises, $y$ is the output of the plant, and $u$ is the control field.}
\label{fig:controller1}
\end{figure}

\section{Active Quantum Controller Design}
\label{sec4}
In this section, we consider another decomposition of the system uncertainty $\Delta A(t)$ for the OPO. If we denote $A$ in \eqref{eqn:systemparameters} as:
$$A=\begin{bmatrix}-\frac{\kappa}{2}&0\\0& -\frac{\kappa}{2}\end{bmatrix}+\chi\begin{bmatrix}1&0\\0&-1\end{bmatrix},$$ 
we obtain
\begin{equation}
\label{eqn:uncertainty_2}
\Delta A_2(t)=\chi \begin{bmatrix}-1+\cos (\Delta \phi(t))&\sin(\Delta \phi(t))\\\sin(\Delta \phi(t))&1-\cos(\Delta \phi(t))\end{bmatrix}.
\end{equation}

We may choose $H_1=\chi\mbox{I}$, $E_1=\mbox{I}$, and 
$$F_2(t)=\begin{bmatrix}-1+\cos (\Delta \phi(t))&\sin(\Delta \phi(t))\\\sin(\Delta \phi(t))&1-\cos(\Delta \phi(t))\end{bmatrix}.$$
The norm-bounded condition can be verified by the following equations:
\begin{equation}
\label{eqn:normF2}
F_2^T(t)F_2(t)=\begin{bmatrix}2-\cos(\Delta \phi(t))&0\\0&2-2\cos(\Delta \phi(t))\end{bmatrix}\leq \rho_2^2 \mbox{I},
\end{equation}
where the bound $\rho_2$ will change with the range of $\Delta \phi(t)$. Particularly, for any $\alpha\in [0,\pi]$ we have 
\begin{equation*}
\begin{aligned}
&\Delta \phi(t)\in [-\alpha, 0], \rho_2\in[2,+\infty),\\
&\Delta \phi(t)\in [0, \alpha], \rho_2\in[\sqrt{2-2\cos(\alpha)},+\infty).
 \end{aligned}
\end{equation*}
In this decomposition, we have $\Delta A_2(t)=0$ when $\Delta \theta (t)=0$. However, the required controller needs to be designed as an active system. Here an active quantum controller means the controller is implemented by using active optical components, such as another OPO.

\subsection{Existence of solutions to Riccati equations}
The robust $H^\infty$ control design theory proposed in Section \ref{sec:controldesigntheory} works directly for the case with uncertainty bound $\rho_2$. However, since the matrix $A$ in this case is different from the case of passive controller, the existence conditions of solutions to Riccati equations \eqref{eqn:simRiccati_a} and \eqref{eqn:simRiccati_b} are different. We have the following proposition.
\begin{proposition}
\label{existenceconditions_case2}
For the OPO in Fig.\ref{fig:plant} with parameters shown in \eqref{eqn:systemparameters} and the system uncertainty shown in \eqref{eqn:uncertainty_2}. If $\kappa_2-\frac{\kappa_1}{\gamma^2}>0$, we can find a positive $\epsilon>0$ such that
\begin{equation}
\label{con:controlable2_1}
    \kappa_2 - \frac{\kappa_1}{\gamma^2 }- \chi^2\epsilon\rho_2^2> 0,
\end{equation}
and the equation \eqref{eqn:simRiccati_a} has a stabilising solution $X$ if and only if:
\begin{equation}
\label{con:nopureimaginaryeigen2_1}
\begin{aligned}
\epsilon\gamma^2\left[(\kappa_1+\kappa_2)^2-4\chi(\kappa_1-\kappa_2)+4\chi^2(1-\rho_2^2)\right]\\
+ 4(\gamma^2\kappa_2-\kappa_1) \geq 0,\\ 
\epsilon\gamma^2\left[(\kappa_1+\kappa_2)^2+4\chi(\kappa_1-\kappa_2)+4\chi^2\
(1-\rho_2^2)\right]\\
+4(\gamma^2\kappa_2-\kappa_1)\geq 0.
\end{aligned}
\end{equation}
If $\kappa_2-\frac{\kappa_1}{\gamma^2}\leq 0$, we can find a positive $\epsilon$ such that
\begin{equation}
\label{con:controlable2_2}
\kappa_1\gamma^2 - \kappa_2 -\frac{1}{\epsilon}>0,
\end{equation}
and the Riccati equation \eqref{eqn:simRiccati_b} has a stabilising solution $Y$ if and only if:
\begin{equation}
\label{con:nopureimaginaryeigen2_2}
\begin{aligned}
4\chi^2\rho_2^2\epsilon(\gamma^2\kappa_1-\kappa_2)+(\kappa_1+\kappa_2)^2-4\chi(\kappa_1-\kappa_2)\\
+4\chi^2(1-\rho_2^2) \geq 0,\\
4\chi^2\rho_2^2\epsilon(\gamma^2\kappa_1-\kappa_2)+(\kappa_1+\kappa_2)^2+4\chi(\kappa_1-\kappa_2)\\
+4\chi^2(1-\rho_2^2) \geq 0.
\end{aligned}
\end{equation}
\end{proposition}

\begin{proof}
Since the form of the matrix $A$ does not change the inequalities $B_1G^{-1}B_1^T-\epsilon \rho^2H_1H_1^T-\gamma^{-2}B_2B_2^T\geq 0$ and $\gamma^2C_2^T\Gamma^{-1}C_2-\frac{1}{\epsilon}E_1^TE_1-C_1^TC_1\geq 0$, we obtain the conditions \eqref{con:controlable2_1} and \eqref{con:controlable2_2} directly from the Proposition \ref{existenceconditions_case1}.

Then by ensuring that the Hamiltonian matrix for the two Riccati equations has no pure imaginary eigenvalues, we obtain \eqref{con:nopureimaginaryeigen2_1} and \eqref{con:nopureimaginaryeigen2_2}, and the details are omitted here.
\end{proof}

\subsection{Numerical example for active controller design}
In this section, we give a numerical example for the same OPO in Section \ref{sec:numericalcalculationcase1}. Fig. \ref{fig:case2possibleepsilon} shows the possible $\epsilon$ for different $\gamma$ and $\rho_2$ to ensure that the Riccati equations \eqref{eqn:simRiccati_a} and \eqref{eqn:simRiccati_b} have stabilising solutions. Three different lines represent the upper bounds on $\epsilon$. The values of $\epsilon$ for a given $\gamma$ and $\rho_2$ can be chosen within the area between the upper bound and the lower bound.
\begin{figure}
\includegraphics[width=0.5\textwidth]{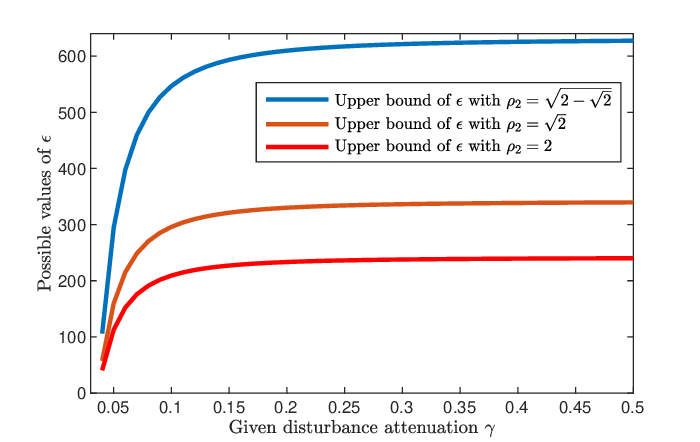}
\caption{Ranges of $\epsilon$ change with disturbance attenuation $\gamma$ for different $\rho_2$.}
\label{fig:case2possibleepsilon}
\end{figure}

For a given disturbance attenuation $\gamma=0.05$, we may choose $\epsilon=1$, and obtain the system parameters as:
\begin{equation} 
\begin{aligned}
A&=\begin{bmatrix}-0.3724&0\\0&-0.4551\end{bmatrix},\\
F_2&=\begin{bmatrix}-1+\cos(\Delta \phi(t))&\sin(\Delta \phi(t))\\ \sin(\Delta \phi(t))&1-\cos(\Delta \phi(t))\end{bmatrix},\\
B_1&=\begin{bmatrix}0.9091&0\\0&0.9091\end{bmatrix}, B_2=\begin{bmatrix}0.0332&0\\0&0.0332\end{bmatrix},\\
C_1&=\begin{bmatrix}0.9091&0\\0&0.9091\end{bmatrix}, C_2=\begin{bmatrix}0.0332&0\\0&0.0332\end{bmatrix},\\
D_1&=-\mbox{I}, D_2=-\mbox{I}.
\end{aligned}
\end{equation}
When the phase uncertainty satisfies $\Delta \phi(t)\in [-\pi, \pi]$, we have $\rho_2=2$ and obtain solutions to the Riccati equations \eqref{eqn:simRiccati_a} and \eqref{eqn:simRiccati_a} as follows:
\begin{equation*}
X=\begin{bmatrix}3.2125&0\\0&2.8733\end{bmatrix}, Y=\begin{bmatrix}0.0094&0\\0&0.0077\end{bmatrix}.\\
\end{equation*}
According to Theorem \ref{the:Riccati}, the controller parameters are:
\begin{equation}
\label{eqn:controller2parameters}
\begin{aligned}
\mathcal{A}_c&=\left[\begin{smallmatrix} -2.2219&0\\0&-2.0109 \end{smallmatrix}\right],\mathcal{B}_c=\left[\begin{smallmatrix} -0.0342&0\\0&-0.0339 \end{smallmatrix}\right],\\
\mathcal{C}_c&=\left[\begin{smallmatrix} -2.0113&0\\0&-1.7030 \end{smallmatrix}\right].
\end{aligned}
\end{equation}

Similarly, we first introduce additional quantum noises to ensure that this active quantum controller is physical realisable. We obtain
\begin{equation}
B_{\nu1}=\begin{bmatrix}2.0113 &0\\ 0 &1.7030\end{bmatrix}, B_{\nu2}=\begin{bmatrix}0.8979&0\\0&0.8979\end{bmatrix}.
\end{equation}
Since the diagonal elements are not the same for matrices $\mathcal{A}_c$, $\mathcal{B}_c$, $\mathcal{C}_c$, and $B_{\nu1}$, the controller will be implemented as a cascade of an OPO and three static squeezers (See \cite{liu2020fault} for details of static squeezers). The structure is shown in Fig. \ref{fig:controller2}. The decay rates for mirrors in the OPO are $\kappa_1, \kappa_2$, and $\kappa_3$, and the pumping-related coefficient is $\chi_c$. The three static squeezers are shown in dashed squares. All these three static squeezers are composed by two fully reflective mirrors and one partially mirror. The three static squeezers are pumped by $\chi_c^{\prime}$, $\chi_c^{\prime \prime}$, and $\chi_c^{\prime \prime}$.
\begin{figure}
\includegraphics[width=0.5\textwidth]{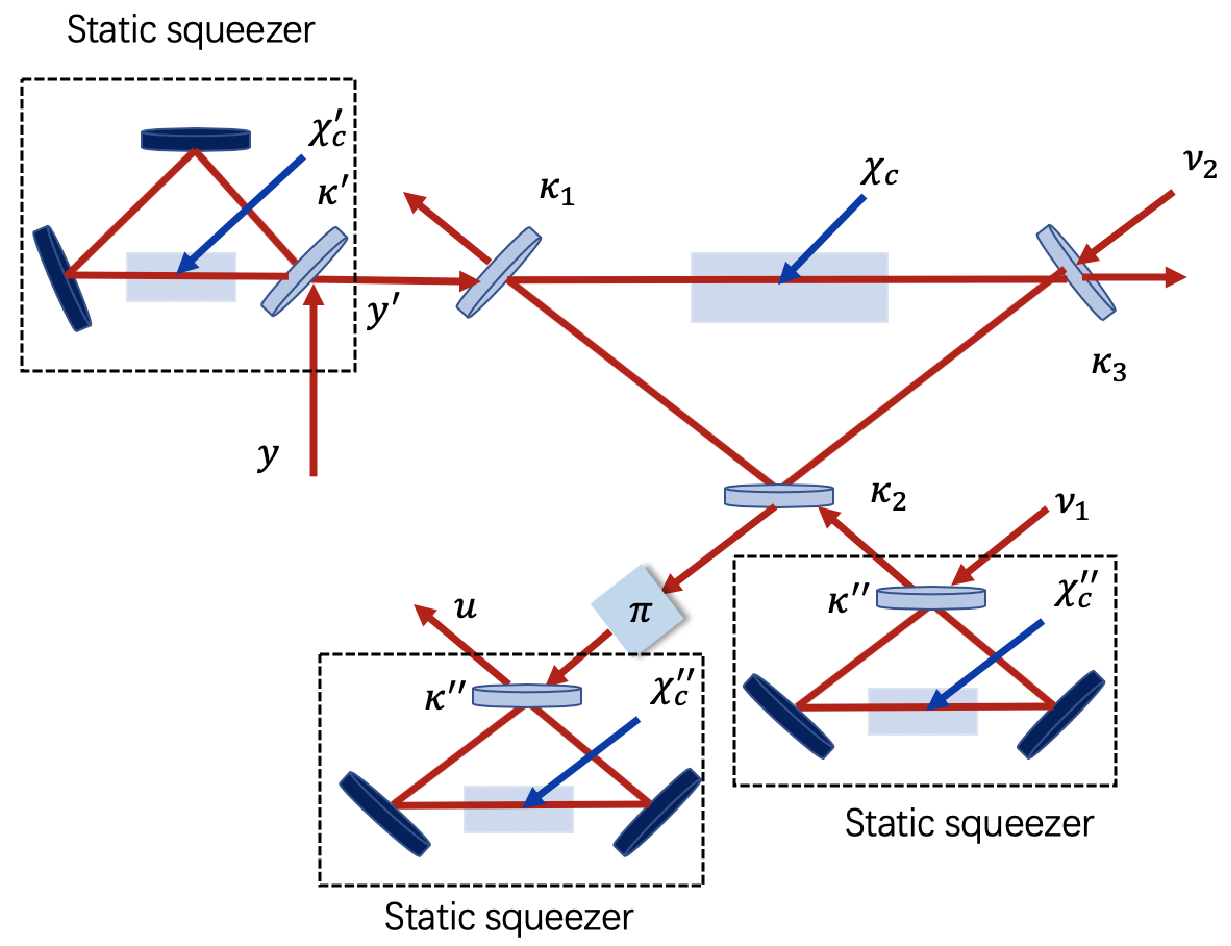}
\caption{The structure of the active quantum controller, which is composed by one OPO, three static squeezers and a phase shifter $\pi$.}
\label{fig:controller2}
\end{figure}
The parameters of this active controller are calculated as follows:
\begin{equation*}
\begin{aligned}
\kappa_1&=0.0006, \kappa_2=3.4260, \kappa_3=0.8979, \chi_c=-0.1055,\\
\kappa^{\prime}&=82.9561, \kappa^{\prime \prime}=5.0784,\\
\chi_c^{\prime}&=\chi_c^{\prime \prime}=0.1055.\\
\end{aligned}
\end{equation*}

\subsection{Comparison of control performance}
\begin{figure}
\includegraphics[width=0.5\textwidth]{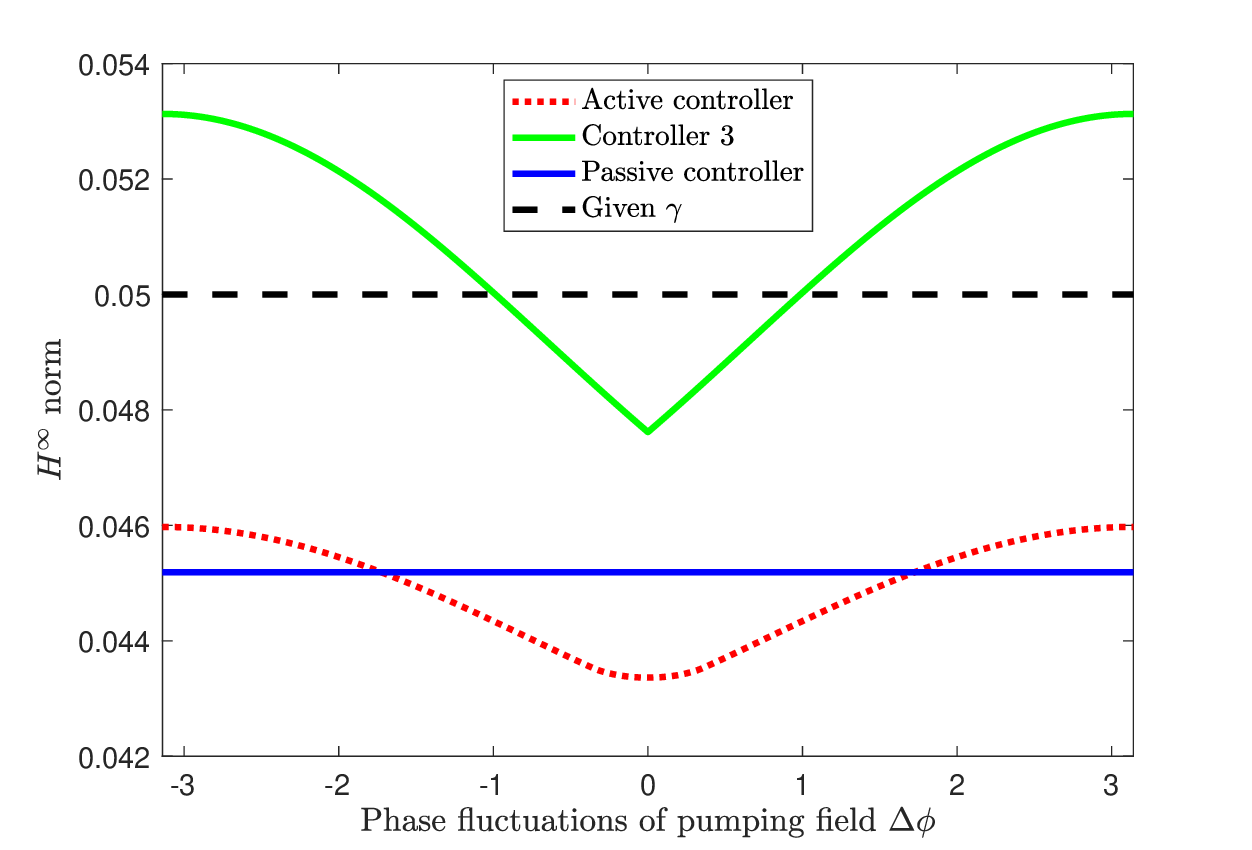}
\caption{Comparison of control performance between three different controllers for a given disturbance attenuation $\gamma=0.05$, $\Delta \phi\in[-\pi,\pi]$, and $\epsilon=1$. \label{fig:performancecomparison}}
\end{figure}

Figure \ref{fig:performancecomparison} shows the $H^\infty$ norm of the transfer function from the disturbance input $\omega(t)$ to the output $z(t)$ for three different controllers. The passive and active controllers are designed in this paper, and Controller 3 is designed without considering uncertainties using the method in \cite{james2008h}.  Fig. \ref{fig:performancecomparison} shows that Controller 3 has good robustness for small phase fluctuations within about $\Delta \phi(t)\in[-1.005,0.9739]$, but is not able to deal with larger fluctuations. Both the passive and active controllers can achieve the required disturbance attenuation $\gamma=0.05$ for all phase fluctuations in $[-\pi, \pi]$, and they also have better performance than Controller 3.

We show that the proposed controllers can achieve the required performance for all $\Delta \phi(t) \in [-\pi, \pi]$ in Fig. \ref{fig:performancecomparison}. It is worth noting that when $\Delta \phi(t)$ is relatively large or small, the performance of the active controller degrades. This may be because the actual bound of $F_2$ in \eqref{eqn:normF2} increases. However, it may be more practical to consider a smaller range of the phase fluctuations in many applications. Fig. \ref{fig:performancecomparison_1} shows the control performance for three controllers with smaller phase fluctuations $\Delta \phi(t) \in [-\frac{\pi}{4}, \frac{\pi}{4}]$. From Fig. \ref{fig:performancecomparison_1}, we can see that different ranges of $\Delta \phi$ will result in different uncertainty bounds $\rho_2$, and then result in different control performance. While the norm of uncertainty  $\rho_1$ is not affected by a smaller range of $\Delta\phi$, and the control performance of the passive controller remains the same. 
\begin{figure}
\includegraphics[width=0.5\textwidth]{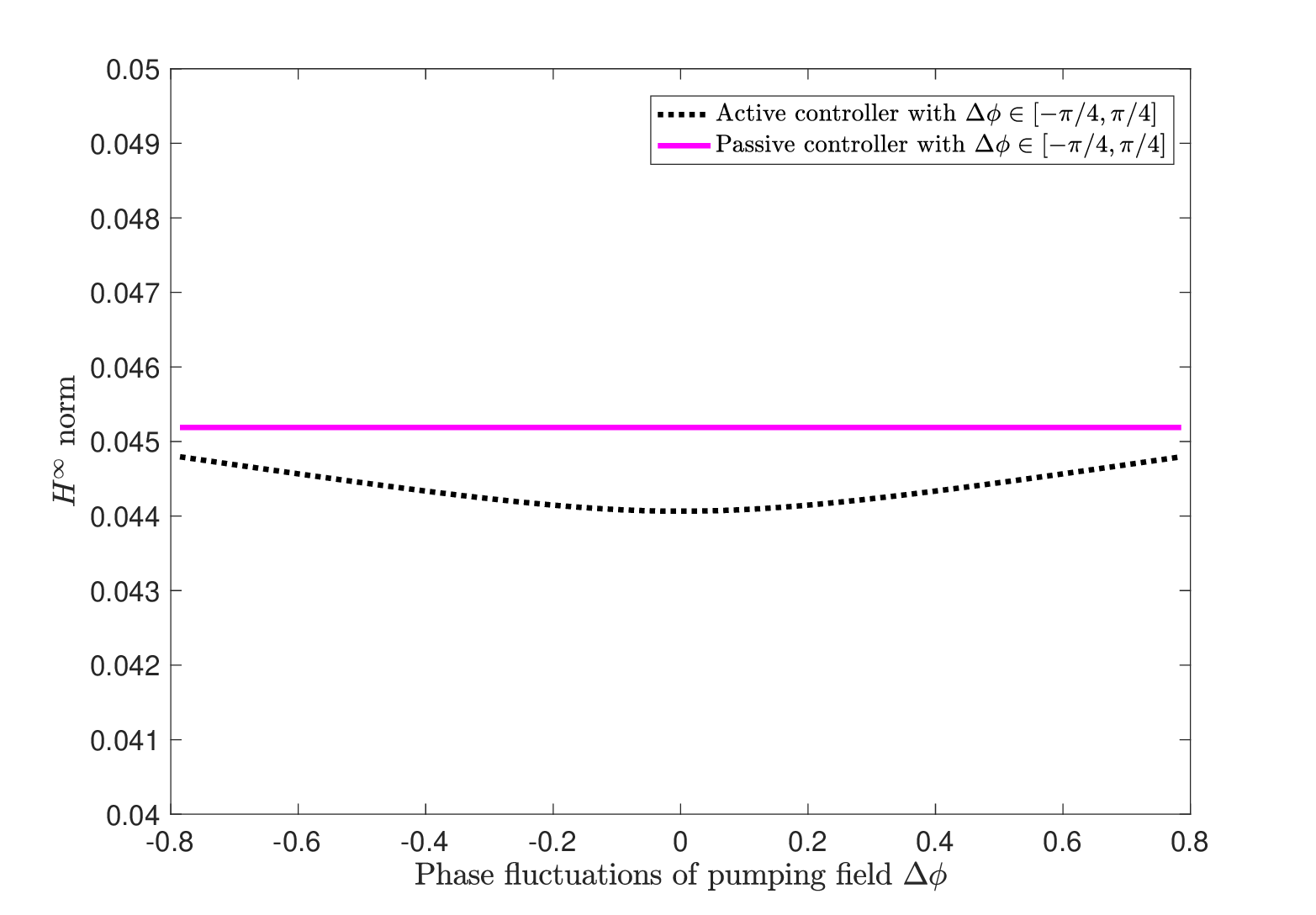}
\caption{Control performance of passive and active controllers between $\Delta \phi\in[-\frac{\pi}{4},\frac{\pi}{4}]$. \label{fig:performancecomparison_1}}
\end{figure}

\section{Fluctuations in the pumping amplitude}
\label{sec5}
In above sections, we have considered fluctuations in the pumping phase. While in some practical applications, there is possibility that the amplitude of the pumping field also suffers from fluctuations. In this section, we show that the proposed two classes of controllers have good robustness with respect to fluctuations in the amplitude of the pumping field. In \cite{liu2020fault}, the authors considered the case that the amplitude of pumping field suffers from a fault signal, and then the system jumps between different modes. Here, we assume the fluctuation of the pumping amplitude is norm-bounded and can be considered as a system uncertainty. We first give the following assumption.
\begin{assumption}
\label{asp:boundofamplitudefluc}
Assume the pumping field can be written as $b=(\beta+\Delta \beta(t))e^{i\Delta \phi(t)}$, where $\Delta \beta(t)$ is the system uncertainty caused by amplitude fluctuations and satisfies $\frac{\Delta \beta(t)}{\beta}\in[0,\beta_{bound}]$. $\beta_{bound}$ represents the bound of ratio between the uncertainty $\Delta \beta(t)$ and the maximum amplitude $\beta$.
\end{assumption}

The fluctuations in the pumping amplitude will affect the norm bounds $\rho_1$ and $\rho_2$ in the design of both a passive controller and an active controller. Hence, we consider the norm-bounded conditions first, and the control procedures of systems with phase fluctuations developed in the previous sections can be directly used for systems with both amplitude and phase fluctuations.
\subsection{Passive quantum controller}
In the design of a passive controller, we have $$\Delta A(t)=\chi\left(1+\frac{\Delta \beta(t)}{\beta}\right) \begin{bmatrix}\cos (\Delta \phi(t))&\sin (\Delta \phi(t))\\\sin (\Delta \phi(t))&-\cos (\Delta \phi(t)) \end{bmatrix}.$$ 
We choose $H_1=\chi\mbox{I}$, $E_1=\mbox{I}$, and 
$$F_1(t)=\left(1+\frac{\Delta \beta(t)}{\beta}\right)\begin{bmatrix}\cos (\Delta \phi(t))&\sin (\Delta \phi(t))\\\sin (\Delta \phi(t))&-\cos (\Delta \phi(t))\end{bmatrix}.$$ 
In this case, $F_1(t)$ satisfies the norm-bounded condition defined by the following inequality:
\begin{equation}
\label{eqn:pboundewithamp}
F_1^T(t)F_1(t)\leq \left(1+\frac{\Delta \beta(t)}{\beta}\right)^2\mbox{I}.
\end{equation} 

Hence, we have the following corollary.
\begin{corollary}
When the bound on the uncertainty $F_1(t)$ satisfies \eqref{eqn:pboundewithamp}, the design procedure for the passive controller and Proposition \ref{existenceconditions_case1} can be directly applied to the case with uncertainties in both amplitude and phase by setting $\rho_1=1+\beta_{bound}$.
\end{corollary}

\subsection{Active quantum controller}
In the design of an active controller, we obtain
\begin{equation*}
\begin{aligned}
\Delta A(t)=&\chi\left(\left[\begin{smallmatrix}-1+\cos(\Delta \phi(t))&\sin(\Delta \phi(t))\\\sin(\Delta \phi(t))&1-\cos(\Delta \phi(t))\end{smallmatrix}\right]\right.\\
&\left.+\frac{\Delta \beta(t)}{\beta}\left[\begin{smallmatrix}\cos (\Delta \phi(t))&\sin (\Delta \phi(t))\\\sin (\Delta \phi(t))&-\cos (\Delta \phi(t))\end{smallmatrix}\right]\right),
\end{aligned}
\end{equation*}
Let $H_1=\chi\mbox{I}$, $E_1=\mbox{I}$, and 
\begin{equation*}
\begin{aligned}
F_2(t)&=\left[\begin{smallmatrix}-1+\cos(\Delta \phi(t))&\sin(\Delta \phi(t))\\\sin(\Delta \phi(t))&1-\cos(\Delta \phi(t))\end{smallmatrix}\right]\\
&+\frac{\Delta \beta(t)}{\beta}\left[\begin{smallmatrix}\cos (\Delta \phi(t))&\sin (\Delta \phi(t))\\\sin (\Delta \phi(t))&-\cos (\Delta \phi(t))\end{smallmatrix}\right].
\end{aligned}
\end{equation*}
The uncertainty $F_2(t)$ satisfies the norm-bounded condition by the following equation:
\begin{equation}
\label{eqn:aboundewithamp}
F_2^T(t)F_2(t)\leq\left(2+\frac{\Delta \beta(t)}{\beta}\right)^2\mbox{I}.
\end{equation}
Similar to the passive case, we obtain the following corollary.

\begin{corollary}
When the bound on the uncertainty $F_2(t)$ satisfies \eqref{eqn:aboundewithamp}, the design procedure for the active controller and Proposition \ref{existenceconditions_case2} in Section \ref{sec4} can be applied here by setting $\rho_2=2+\beta_{bound}$.
\end{corollary}

In practical applications, we may assume e.g., $\beta_{bound}\leq5\%$. The passive and active controllers can be designed in a similar way to the case with only phase fluctuations. We omit the numerical calculations of the Riccati equations and controller parameters, and only show simulation results of $H^\infty$ norm. The results are presented in Fig. \ref{fig:performancecomparisonwithamplityde}, which shows a similar result to that in Fig. \ref{fig:performancecomparison}. The curves show some fluctuations caused by the amplitude fluctuations. In the simulation result, we take $\frac{\Delta\beta(t)}{\beta}$ as a random value in $[0,5\%]$ with each $\Delta\phi$. The two controllers designed in this paper achieve the given disturbance attenuation $\gamma=0.05$.
\begin{figure}
\includegraphics[width=0.5\textwidth]{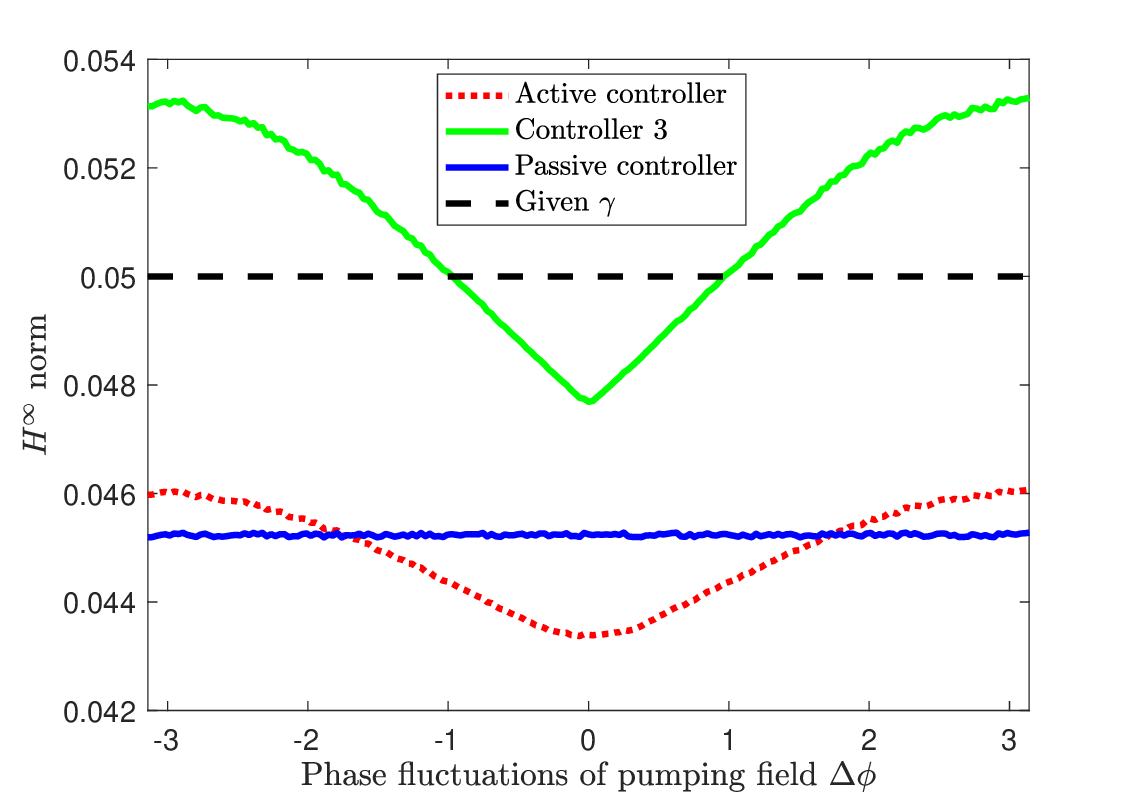}
\caption{Comparison of control performance between three different controllers of the system with fluctuations both in amplitude and phase, for a given disturbance attenuation $\gamma=0.05$, $\Delta \phi\in[-\pi,\pi]$, $\frac{\Delta \beta(t)}{\beta(t)}\in[0, \beta_{bound}]$, $\beta_{bound}=5\%$, and $\epsilon=1$. \label{fig:performancecomparisonwithamplityde}}
\end{figure}

\section{Conclusions}
\label{sec:conclusion}
In this paper, we considered robust $H^\infty$ control design for an OPO composed of a cavity and a nonlinear crystal, where the phase of its pumping field is subject to fluctuations due to a fault signal. The controller was designed by solving two Riccati equations, and the physical stabilisation of the designed controller was theoretically discussed. A passive controller was first designed, which has a simple structure, to achieve the given control performance. The passive controller was implemented by an empty cavity that is composed of three mirrors. Based on another decomposition of the system uncertainties caused by phase fluctuations, an active controller was then designed to achieve the required control performance. The implementation of the active controller is more complex, and it is composed by static squeezers and an OPO, while it has better performance than the passive controller. This paper also compared the control performance of the proposed two controllers and the controller designed without consideration of the system uncertainties. The results show that the proposed two controllers achieve better performance with respect to the $H^\infty$ norm. We also illustrated that the two controllers proposed have good robustness to amplitude fluctuations in the pumping field. Future work may include performance improvement of output squeezing for OPOs using the proposed approach, establishing sufficient and necessary conditions for the existence of solutions to the two Riccati equations for a more general case and a comprehensive performance comparison between passive controllers and active controllers.



\bibliographystyle{apacite}  
\bibliography{Automatica_fault_tolerant}

\end{document}